\documentclass[twocolumn,aps,longbibliography,superscriptaddress,9pt]{revtex4-1}
\usepackage[utf8]{inputenc}
\usepackage{graphicx}
\usepackage{verbatim}
\usepackage{braket}
\usepackage{siunitx}
\usepackage{chemformula}

\usepackage{amsmath}
\usepackage{amsfonts}
\usepackage{amssymb}
\usepackage{amsthm}
\usepackage{mathtools}
\usepackage{nicefrac}
\DeclareMathOperator{\cnot}{CNOT}

\newcommand{\I}{\textnormal{\texttt{i}}}
\newcommand{\conj}[1]{\overline{#1}}

\usepackage[colorlinks,urlcolor=blue,linkcolor=black,citecolor=blue]{hyperref}
\usepackage[capitalize,nameinlink]{cleveref}

\newtheorem{thrm}{Theorem}
\newtheorem{definition}[thrm]{Definition}
\newtheorem{lemma}[thrm]{Lemma}
\newtheorem{conjecture}[thrm]{Conjecture}

\usepackage{tabularx}
\newcolumntype{L}[1]{>{\hsize=#1\hsize\raggedright\arraybackslash}X}%
\newcolumntype{R}[1]{>{\hsize=#1\hsize\raggedleft\arraybackslash}X}%
\newcolumntype{C}[1]{>{\hsize=#1\hsize\centering\arraybackslash}X}%

\usepackage{tikz}
\usepackage{soul}

\usepackage{xcolor}

\usepackage{qcircuit}
\usepackage{environ}

\newcommand{\myqctmp}[2][0.75]{\Qcircuit @C=#2em @R=#1em @!R}
\NewEnviron{myqcircuit}[1][0.75]{\vcenter{\myqctmp[#1]{0.8} {\BODY}}}
\NewEnviron{myqcircuit*}[2]{\vcenter{\myqctmp[#1]{#2} {\BODY}}}

\newcommand{\MG}{\multigate1{G}}
\newcommand{\mg}{\ghost{G}}
\newcommand{\gw}{\push{\rule{0em}{2ex}}\qw}
\newcommand{\MGG}[1]{\multigate1{G(\Theta_{#1})}}
\newcommand{\mgg}[1]{\ghost{G(\Theta_{#1})}}

\newcommand{\Rxx}[1]{\gate{R_x(\theta_{#1})}}
\newcommand{\Ryy}[1]{\gate{R_y(\theta_{#1})}}
\newcommand{\Rzz}[1]{\gate{R_z(\theta_{#1})}}
\newcommand{\Rxp}{\gate{R_x(\nicefrac\pi2)}}
\newcommand{\Rxm}{\gate{R_x(\nicefrac{-\pi}2)}}
\newcommand{\Rzp}{\gate{R_z(\nicefrac\pi2)}}
\newcommand{\Rzm}{\gate{R_z(\nicefrac{-\pi}2)}}

\newcommand{\qdots}{\push{~\cdots~}\qw}

\newcommand{\bA}{\textcolor{red} {\boldsymbol a}}
\newcommand{\bB}{\textcolor{blue}{\boldsymbol b}}
\newcommand{\bC}{\textcolor{cyan}{\boldsymbol c}}
\newcommand{\bD}{\textcolor{cyan}{\boldsymbol d}}
\newcommand{\bE}{\textcolor{red} {\boldsymbol e}}
\newcommand{\bF}{\textcolor{blue}{\boldsymbol f}}
\newcommand{\bG}{\textcolor{cyan}{\boldsymbol g}}
\newcommand{\bH}{\textcolor{cyan}{\boldsymbol h}}

\newcommand{\Hcd}{\mathcal{H_{CD}}}

\begin{document}

\title{Constant-Depth Circuits for Dynamic Simulations of Materials on Quantum Computers}

\author{Lindsay Bassman}
\affiliation{Lawrence Berkeley National Lab, Berkeley, CA 94720}
\author{Roel Van Beeumen}
\affiliation{Lawrence Berkeley National Lab, Berkeley, CA 94720}
\author{Ed Younis}
\affiliation{Lawrence Berkeley National Lab, Berkeley, CA 94720}
\author{Ethan Smith}
\affiliation{University of California Berkeley, Berkeley, CA 94720}
\author{Costin Iancu}
\affiliation{Lawrence Berkeley National Lab, Berkeley, CA 94720}
\author{Wibe A. de Jong}
\affiliation{Lawrence Berkeley National Lab, Berkeley, CA 94720}

\begin{abstract}
Dynamic simulation of materials is a promising application for near-term quantum computers.  Current algorithms for Hamiltonian simulation, however, produce circuits that grow in depth with increasing simulation time, limiting feasible simulations to short-time dynamics.  Here, we present a method for generating circuits that are constant in depth with increasing simulation time for a subset of one-dimensional materials Hamiltonians, thereby enabling simulations out to arbitrarily long times. Furthermore, by removing the effective limit on the number of feasibly simulatable time-steps, the constant-depth circuits enable Trotter error to be made negligibly small by allowing simulations to be broken into arbitrarily many time-steps.  Composed of two-qubit matchgates on nearest-neighbor qubits, these constant-depth circuits are constructed based on a set of multi-matchgate identity relationships.  For an $N$-spin system, the constant-depth circuit contains only $\mathcal{O}(N^2)$ $\cnot$ gates.  When compared to standard Hamiltonian simulation algorithms, our method generates circuits with order-of-magnitude fewer gates, which allows us to successfully simulate the long-time dynamics of systems with up to 5 spins on available quantum hardware.  This paves the way for simulations of long-time dynamics for scientifically and technologically relevant quantum materials, enabling the observation of interesting and important atomic-level physics.
\end{abstract}

\maketitle
\section{Introduction}
While a quantum advantage was recently achieved with random circuits \cite{arute2019quantum}, it remains a challenge to demonstrate a quantum advantage for an application of interest within the physical sciences, a feat which has been dubbed ``physical quantum advantage''.  This is because current and near-term quantum computers, otherwise known as noisy intermediate-scale quantum (NISQ) computers \cite{preskill2018quantum}, have low qubit counts and suffer from short qubit decoherence times and high gate error rates, making it difficult to perform relevant, large-scale computations.  Given such constraints, long-anticipated applications like number factorization \cite{shor1999polynomial} and unordered database search \cite{grover1996fast} are still far out of reach for NISQ computers.  Quantum computers, however, are intrinsically fit for efficiently simulating quantum systems \cite{feynman1982simulating, lloyd1996universal, abrams1997simulations, zalka1998simulating}, making the dynamic simulation of quantum materials a leading ``killer application" for NISQ computers.  Rapid progress in both quantum hardware and software may soon allow for such simulations to not only demonstrate a physical quantum advantage, but to advance such fields as condensed matter physics, quantum chemistry, and materials science.  

One of the major challenges with performing dynamic simulations on NISQ devices is keeping the quantum circuits small enough to produce high-fidelity results.  Dynamic simulations require the execution of one circuit per time-step, where each circuit implements the time-evolution operator from the initial time to the given time-step \cite{bassman2021simulating}.  Current algorithms for dynamic materials simulations produce quantum circuits whose depths grow with increasing simulation time-steps \cite{wiebe11, childs2018toward}.  Thus, an essential part of the workflow for simulating the dynamics of materials on NISQ computers is quantum circuit optimization, which can minimize the depth of the circuits produced by current algorithms.  Already, a great deal of research has focused on general circuit optimization (i.e. minimization) \cite{mottonen2004quantum, de2016block, iten2016quantum, martinez2016compiling, khatri2019quantum, murali2019noise, younis2020qfast, cincio2020machine}, and domain-specific circuit optimizers, which focus on optimizing certain types of circuits for specific applications, have been suggested \cite{bassman2020domain} as a method to reduce to complexity of the optimization problem, which in general is NP-hard \cite{botea2018, herr17}.  

According to the ``no-fast-forwarding theorem'', simulating the dynamics of a system under a generic Hamiltonian $H$ for a time $t$ requires $\Omega(t)$ gates \cite{berry2007efficient, childs2009limitations}, implying that circuit depths grow at least linearly with the number of time-steps.  It has been shown, however, that quadratic Hamiltonians can be fast forwarded, meaning the evolution of the systems under such Hamiltonians can be simulated with circuits whose depths do not grow significantly with the simulation time \cite{atia2017fast}.  A recent work took advantage of this to variationally compile approximate circuits with a hybrid classical-quantum algorithm for fast-forwarded simulations \cite{cirstoiu2020variational}.  The circuits, however, are approximate, with error that grows with increasing fast-forwarding time.

Here, we present an algorithm for producing quantum circuits that are constant in depth with increasing time-step count for simulations of materials governed by special models derived from the Heisenberg Hamiltonian, which we denote by $\Hcd$ and define in \cref{sec_theory}.  The constant-depth circuits have a fixed structure, with only the single-qubit rotation angles changing with the addition of more time-steps.  The structure has a regular pattern which can be easily extrapolated to build circuits for any system size; for an $N$-spin system, the circuit structure contains only $N(N-1)$ $\cnot$ gates.  Furthermore, the circuits are exact up to Trotter error, which we argue, can be practically eliminated.  This is because Trotter error scales with the size of the simulation time-step, and the constant-depth nature of the circuits allows for a simulation to be feasibly broken into arbitrarily many (i.e., arbitrarily small) time-steps.

The circuits are comprised of two-qubit gates, known as matchgates \cite{valiant2002quantum}.  While generic matchgates decompose into native-gate circuits with three $\cnot$ gates \cite{vidal2004universal}, the matchgates for Hamiltonians in $\Hcd$, have the special property that they only require two $\cnot$ gates in their decomposition.  This special property allows us to introduce a set of conjectured matchgate identities, which enable the downfolding of our circuits for dynamic simulations into constant-depth for any number of time-steps.  For a given system size, if the fixed circuit structure is small enough to achieve high-fidelity results on a NISQ computer, the dynamics of that system can be successfully simulated out to arbitrarily long times.  By removing the limit on the number of simulation time-steps that can feasibly be simulated, these constant-depth circuits allow for long-time dynamic simulations, with minimized Trotter error, that can give insights into complex molecular reactions, transformations, and equilibration.

\section{Theoretical Background}
\label{sec_theory}
The quantum circuits for dynamic simulations of quantum materials must implement the time-evolution operator between the initial time (which we set to 0) and some final time $t$, given by
\begin{equation}
    U(0,t)\equiv U(t)=\mathcal{T} \exp(-i\int_{0}^{t}H(t)dt)
\end{equation}
where $\mathcal{T}$ indicates a time-ordered exponential and $H(t)$ is the time-dependent Hamiltonian of the material.  In general, this operator is challenging to compute exactly due to the time-dependence of the Hamiltonian and the exponentiation the Hamiltonian.  First, an approximation must be made which transforms $H(t)$ into a piece-wise function by discretizing time into small time-steps over which $H(t)$ is constant \cite{poulin2011quantum}.  For small system sizes, it is then possible to compute $U(t)$ by exact diagonalization of the Hamiltonian, however, this task becomes exponentially harder with increasing system size.  Thus for larger system sizes, a second approximation must be made to exponentiate the Hamiltonian.  Typically, the Trotter decomposition \cite{trotter1959product} is used, which splits the Hamiltonian into components that are each easy to diagonalize.  With these two approximations, the time-evolution operator becomes
\begin{equation} \label{unitary}
 U(n\Delta t) = \prod_{\tau=1}^{n} \prod_{l} e^{-iH_l(t_{\tau})\Delta t} + \mathcal{O}(\Delta t)  
\end{equation}
where $\tau$ multiplies over the number of discretized time-steps $\Delta t$ and $l$ multiplies over the components into with $H(t)$ was divided.  

The error generated from the Trotter decomposition, known as Trotter error, can be a significant source of error, scaling with the size of the simulation time-step $\Delta t$.  Dynamic simulations based on standard Hamiltonian simulation algorithms must strike a balance when selecting the size of $\Delta t$.  This is because standard algorithms produce circuits which grow in depth with increasing numbers of time-steps, which in turn limits the number of time-steps that are feasible to simulate to just a handful \cite{smith2019simulating}.  While decreasing $\Delta t$ will lower Trotter error, making $\Delta t$ too small will not allow for a long enough total simulation time, since the number of time-steps is limited.  Our constant-depth circuits, however, remove the limitation on the number of time-steps that can be feasibly simulated, since the circuits do not get deeper with higher time-step count.  This allows for the time-step $\Delta t$ to be made arbitrarily small, which in turn allows one to decrease the Trotter error to negligible values.  Such practical elimination of Trotter error with constant-depth circuits can enable far more accurate simulation results for long-time dynamic simulations.

The constant-depth circuits we introduce here simulate the dynamical evolution of a quantum material whose Hamiltonian is a simplified version of the one-dimensional (1D) Heisenberg model, as explained below.  The Heisenberg Hamiltonian is defined as
\begin{equation}
H(t) = -\sum_{\alpha}\{J_{\alpha}\sum_{i=1}^{N-1} \sigma_{i}^{\alpha}\sigma_{i+1}^{\alpha}\} - h_{\beta}(t) \sum_{i=1}^{N} \sigma_{i}^{\beta}
\label{Hamiltonian}
\end{equation}
where $\alpha$ sums over $\{x,y,z\}$, the coupling parameters $J_{\alpha}$ denote the exchange interaction between nearest-neighbor spins along the $\alpha$-direction, $\sigma_{i}^{\alpha}$ is the $\alpha$-Pauli matrix acting on qubit $i$, and $h_{\beta}(t)$ is the time-dependent amplitude of an external magnetic field along the $\beta$-direction, where $\beta \in \{x,y,z\}$.  This Hamiltonian is thus defined by the set of its parameters $\{J_x, J_y, J_z, h_\beta(t)\}$.  We denote the set of all parameter sets as $\mathcal{H}$.  The full Heisenberg model is obtained when all parameters in the set are non-zero, however a number of ubiquitous models can be derived by setting various parameters to zero.

\Cref{tab:const-depth} shows all subsets $\Hcd$ of $\mathcal{H}$ for which we find that our constant-depth circuits work.  The rows of the table denote either the direction of the external magnetic field $h_{\beta}$ or a lack of field, while the columns label which of the coupling parameters are non-zero.  The first three columns denote parameter sets where one coupling term is non-zero, the next three columns denote sets where two coupling terms are non-zero, while the final column denotes the sets where all three coupling parameters are non-zero.  An $\times$ appears in table entries for parameter sets that define Hamiltonians in $\Hcd$, which can be simulated with our constant-depth circuits. Note that $J_x \cdot J_y \cdot J_z = 0$ is a necessary but not sufficient condition for constant-depth.
\begin{table}
\centering
\begin{tabularx}{\columnwidth}{C{.3}|C{.6}C{.6}C{.6}|C{1.3}C{1.3}C{1.3}|C{2}}
                & $J_x$ & $J_y$ & $J_z$ & $J_x+J_y$ & $J_x+J_z$ & $J_y+J_z$ & $J_x+J_y+J_z$ \\ \hline
  $x$           & $\times$&$\times$&$\times$ &         &        &$\times$ \\
  $y$           & $\times$&$\times$&$\times$ &         &$\times$&         \\ 
  $z$           & $\times$&$\times$&$\times$ & $\times$&        &         \\ 
  $\varnothing$ & $\times$&$\times$&$\times$ & $\times$&$\times$&$\times$ \\
\end{tabularx}
\caption{Subsets of Heisenberg parameters $\Hcd$ for which circuits are constant-depth.  The rows denote the direction of the external field or a lack of a field. The columns denote the non-zero coupling parameters.  Table entries marked with an $\times$ denote parameter sets that represent Hamiltonians for which our constant-depth circuits work. }
\label{tab:const-depth}
\end{table}
As all $\Hcd$ Hamiltonians of \cref{tab:const-depth} are quadratic, it is possible to fast-forward simulations under their time-evolution \cite{atia2017fast}.  In \cref{sec_results}, we demonstrate simulations with our constant-depth circuits for two important models in $\Hcd$: (i) the XY model, where $J_z = 0$ and $h_\beta = 0$, and (ii) the transverse field Ising model (TFIM), where $J_y = J_z = 0$.  Dynamics of these models have recently been simulated on quantum computers, but lack of constant-depth circuits either limited the number of time-steps that could be successfully simulated \cite{smith2019simulating}.

\section{Construction of Constant-Depth Circuits}
To arrive at the circuit structure for the constant-depth circuits, we begin by laying down the gates that implement evolution of the system by one time-step, $U(\Delta t)$.  Due to the quadratic nature of $\Hcd$ Hamiltonians, which only contain coupling interactions between nearest neighbor spins, the circuit for evolution of one time-step can be constructed by a set of two-qubit gates which act on each of the pairs of nearest neighbor qubits.  For example, for six qubits, this circuit is given by
\begin{equation}
  U(\Delta t) = {\small\begin{myqcircuit}[0.75]
  & \MGG1 & \qw   & \qw \\
  & \mgg1 & \MGG2 & \qw \\
  & \MGG3 & \mgg2 & \qw \\
  & \mgg3 & \MGG4 & \qw \\
  & \MGG5 & \mgg4 & \qw \\
  & \mgg5 & \qw   & \qw \\
\end{myqcircuit}}
\label{eq:timestep1}
\end{equation}
where each gate labeled $G(\Theta_i)$ is a two-qubit gate defined by some set of parameters $\Theta_i$.  For ease of notation, the parameter set $\Theta_i$ is dropped in subsequent labeling of these gates, which will simply be labeled with a $G$. However, it must be emphasized that each two-qubit gate has its own unique parameter set.  Each additional time-step requires one additional repetition of the circuit for one time-step.  In this manner, it is possible to construct circuits for dynamic simulations that grow with increasing numbers of time-steps.  We refer to these circuits has the ``naive circuits''  for dynamic simulations.  The naive circuit for $n$ time-steps for six qubits is thus given by 
\begin{equation}
  U(n\Delta t) = {\small\begin{myqcircuit}[0.5]
  & \MG & \qw   & \qdots & \MG & \qw   & \qw \\
  & \mg & \MG & \qdots & \mg & \MG & \qw \\
  & \MG & \mg & \qdots & \MG & \mg & \qw \\
  & \mg & \MG & \qdots & \mg & \MG & \qw \\
  & \MG & \mg & \qdots & \MG & \mg & \qw \\
  & \mg & \qw   & \qdots & \mg & \qw   & \qw \\
\end{myqcircuit}}
\label{eq:timestep_n}
\end{equation}
where there are $2n$ columns of $G$ gates for $n$ time-steps.  We now show that it is possible to reduce the naive circuits for higher numbers of time-steps down to a constant-depth circuit which is comprised of $N$ columns of $G$ gates for an $N$-spin system, where each column alternates placing the top of the first $G$ gate on the first or second qubit.  The ability to ``downfold'' longer circuits into constant-depth circuits is derived from special properties of these $G$ gates.

In fact, the $G$ gates belong to a special group of two-qubit gates known as matchgates \cite{valiant2002quantum}.
\begin{definition}
\label{def:matchgate}
Let the matrices $A$ and $B$ be in $SU(2)$
\begin{align}
A &= \begin{bmatrix} p & q \\ r & s \end{bmatrix}, &
B &= \begin{bmatrix} w & x \\ y & z \end{bmatrix},
\end{align}
with $\det(A) = \det(B)$.
Then the two-qubit matchgate $G(A,B)$ is defined as follows
\begin{equation}
G(A,B) =
\begin{bmatrix}
p &   &   & q \\
  & w & x \\
  & y & z \\
r &   &   & s
\end{bmatrix}.
\label{eq:G}
\end{equation}
\end{definition}
Matchgates have the important property that the product of two matchgates is again a matchgate and will be a key feature to arrive at constant-depth circuits.

\begin{lemma}
\label{lem:prod-match}
Let $G(A_1,B_1)$ and $G(A_2,B_2)$ be matchgates, then the matrix
\begin{equation}
G(A_3,B_3) = G(A_1,B_1) \, G(A_2,B_2),
\end{equation}
is again a matchgate with $A_3 = A_1 A_2$ and $B_3 = B_1 B_2$.
\end{lemma}
\begin{proof}
The proof directly follows from carrying out the matrix-matrix multiplication.
\end{proof}
A graphical representation of \cref{lem:prod-match} is given by
\begin{equation}\small
\begin{myqcircuit}
  & \MG & \MG & \qw \\
  & \mg & \mg & \qw
\end{myqcircuit}
\ \equiv \
\begin{myqcircuit}
  & \MG & \qw \\
  & \mg & \qw
\end{myqcircuit}
\label{eq:equiv2}
\end{equation}

The decomposition of a general matchgate into a native-gate circuit, which can be executed on NISQ devices, requires three $\cnot$ gates \cite{vidal2004universal}.  However, all the matchgates for $\Hcd$ have a special structure which allows them to be decomposed into native-gate circuits with only two $\cnot$ gates.  In $\Hcd$ cases with an external magnetic field along the $x$- or $y$-directions, the gates $G$ in \eqref{eq:timestep1} do not have the matchgate structure but are spectrally equivalent with a matchgate and are in fact matchgates up to some $\nicefrac{\pi}{2}$ rotations.  Matchgates and their corresponding decomposition into native-gate circuits with two-$\cnot$s are given in \cref{app:matchgates}. 

The ability to contract the circuits to constant depth relies on an identity that we conjecture for these special $\Hcd$ matchgates.

\begin{conjecture}
\label{conj:vee-hat}
Let $G_1,G_2,G_3$ be matchgates of a certain type in $\Hcd$, then there exist three corresponding matchgates $G_4,G_5,G_6$ of the same type so that
\begin{equation}
(G_1 \otimes I)(I \otimes G_2)(G_3 \otimes I)
= (I \otimes G_4)(G_4 \otimes I)(I \otimes G_6).
\end{equation}
\end{conjecture}

A graphical representation of \cref{conj:vee-hat} is given by
\begin{equation}\small
\begin{myqcircuit}
  & \MG & \qw & \MG & \qw \\
  & \mg & \MG & \mg & \qw \\
  & \qw & \mg & \qw & \qw
\end{myqcircuit}
\ \equiv \
\begin{myqcircuit}
  & \qw & \MG & \qw & \qw \\
  & \MG & \mg & \MG & \qw \\
  & \mg & \qw & \mg & \qw
\end{myqcircuit}
\label{eq:equiv3}
\end{equation}
Using numerical optimization to identify circuit parameters on either side of the equality, we empirically find this conjecture to be true for all trials where the circuits are comprised of $\Hcd$ matchgates.  It has, however, proven challenging to analytically compute the parameters of the circuit on the right-hand side given the parameterized circuit on the left-hand side and vice versa.  As a result, compilation of our constant-depth circuits requires numerical optimization to obtain circuit parameters.  We emphasize that the equivalence \eqref{eq:equiv3} only holds for $\Hcd$ matchgates, whereas the equivalence \eqref{eq:equiv2} holds for all matchgates.

Based on \cref{eq:equiv2,eq:equiv3} we can derive identities for higher numbers of qubits, where a set of $N$ columns of matchgates across $N$ qubits can be replaced by its mirror image, albeit with altered parameter sets for all the constituent matchgates.  A demonstration of deriving the identity for 4-qubits is shown in \cref{app:proof}.  We refer to these identities as the matchgate mirroring identities. These conjectured identities are depicted in \cref{fig:mirror_id} for four and five qubits.  Note that for even numbers of qubits the mirroring is about a vertical axis (\cref{fig:mirror_id}a), while for odd numbers of qubits the mirroring is about a horizontal axis (\cref{fig:mirror_id}b).  We emphasize that the matchgate parameters are different on either side of the equality signs.

\begin{figure}
	\centering
	\includegraphics[scale=0.5]{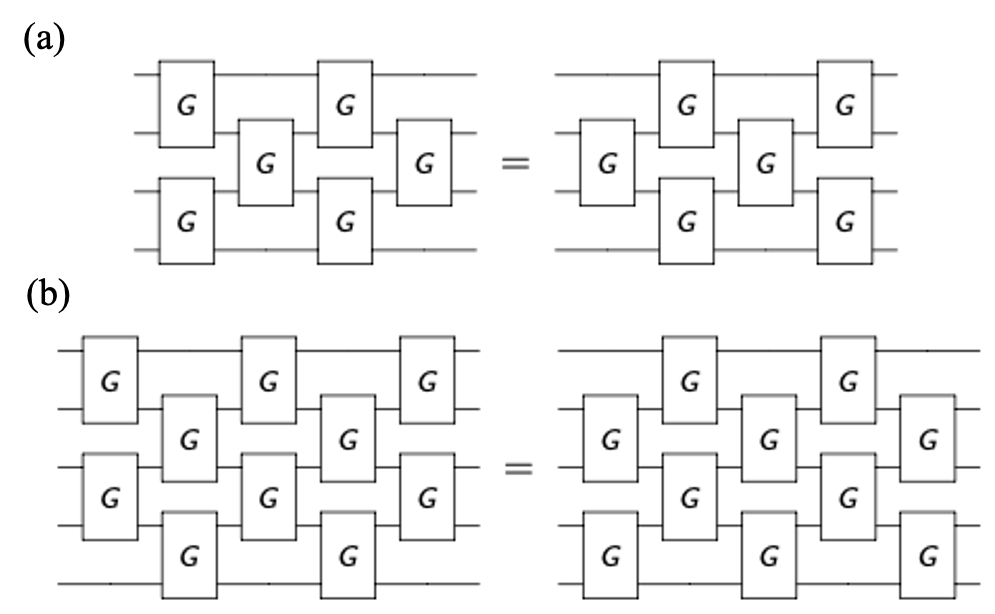}
	\caption{Conjectured $\Hcd$ matchgate mirroring identities for four qubits (a) and for five qubits (b).}
	\label{fig:mirror_id}
\end{figure}

To understand how these mirroring identities allow for the construction of constant depth circuits, we notice that for an $N$-qubit naive circuit (e.g., \cref{eq:timestep_n}) we can apply the matchgate mirroring identity to the last $N$ columns of matchgates in the circuit.  We emphasize that applying this identity will change the parameters defining each matchgate within the mirroring group.  Application of this identity will result in pairs of adjacent matchgates on the same qubit pairs that can be combined into one matchgate, thus reducing the number of columns of matchgates in the circuit by one.  This can be repeated until only $N$ columns of matchgates in the circuit remain. This process is demonstrated for six qubits in \cref{fig:six_qubit}.  \cref{fig:six_qubit}a shows the naive circuit for six qubits simulating $n$ time-steps with the last six columns of matchgates in the circuit highlighted with an outline.  \cref{fig:six_qubit}b shows one application of the matchgate mirroring identity for six qubits to these last six columns of matchgates.  Note how after applying the identity, two pairs of matchgates emerge adjacent to one another on the same pair of qubits, highlighted with an outline.  These pairs can each be merged into one matchgate with new parameters, thus reducing the number of columns of matchgates in the circuit by one.  This process is repeated until only six columns of matchgates remain, as shown in \cref{fig:six_qubit}c.

\begin{figure}
	\centering
	\includegraphics[scale=0.5]{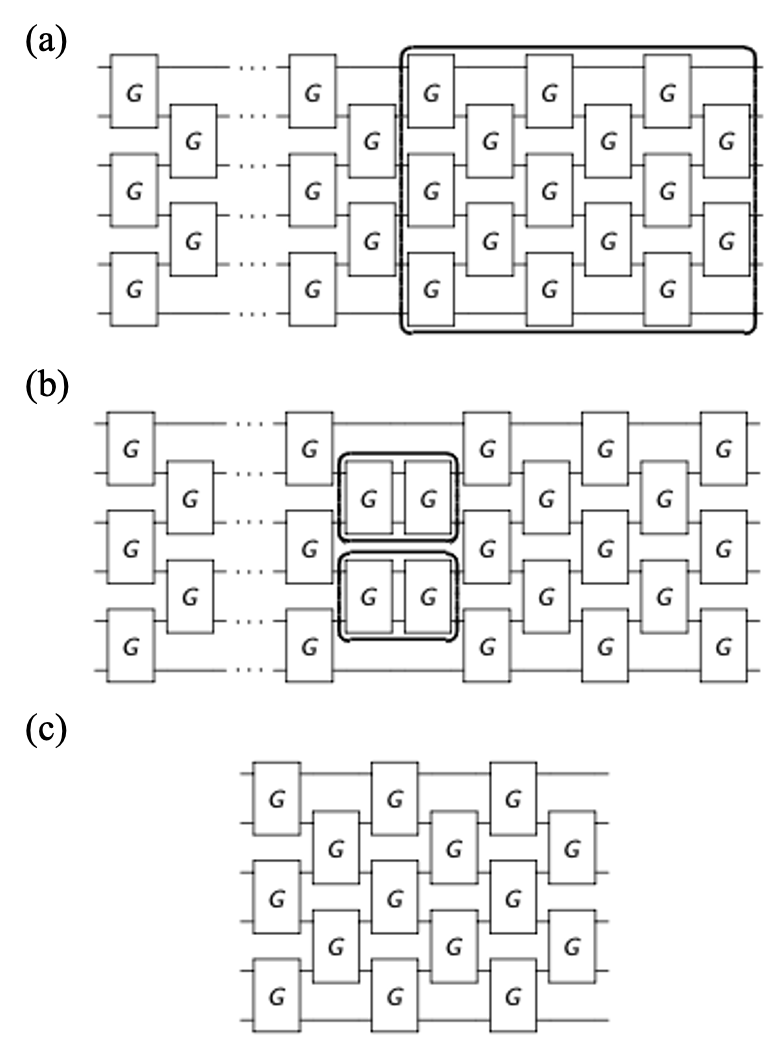}
	\caption{Downfolding a 6-qubit circuit for $n$ time-steps down to a constant-depth circuit. (a) The 6-qubit circuit for evolving the system by $n$ time-steps with the time-evolution operator $U(n\Delta t)$.  A box highlights the last six columns of matchgates to which the matchgate mirroring identity will be applied.  (b) The 6-qubit circuit after application of the $\Hcd$ matchgate mirroring identity.  Pairs of adjacent matchgates on the same qubit pairs which can be combined into one matchgate with new parameters are highlighted with an outline. (c) The final constant-depth circuit for a 6-qubit circuit, which has six columns of matchgates.}
	\label{fig:six_qubit}
\end{figure}

The downfolding approach presented in \cref{fig:six_qubit} shows how to methodically obtain constant-depth circuits for each time-step in the dynamic simulation. In practice, however, we directly use numerical optimization to find the parameters for the constant-depth circuit of \cref{fig:six_qubit}c. We begin by computing the operator in \cref{unitary} (either by exact diagonalization or Trotter decomposition), which defines our target matrix (i.e., the matrix our circuit aims to carry out).  Given a system size, we then construct the constant-depth circuit structure, which has $N$ columns of matchgates for an $N$-qubit system.  Next, we compute the matrix equivalent of the circuit, which will be compared to our target matrix.  Using numerical optimization, we then solve for the parameters of the circuit that minimize the distance between the circuit matrix and the target matrix.

The number of circuit parameters grows quadratically with system size.  This makes scaling to larger system sizes challenging as the circuit optimization for each time-step will take longer to compute.  This could be ameliorated by finding a way to map the coefficients of the Hamiltonian directly to the rotation angles in the constant-depth circuit, whether through analytical techniques or machine learning methods.  This would enable one to skip computation of the time-evolution operator and numerical optimization altogether.  It should be noted, however, that the inability to remove this classical optimization step may not completely inhibit this method because the constant-depth circuit generation is embarrassingly parallel.  In other words, the circuits for each time-step may all be computed in parallel, as numerical optimization for one circuit does not depend on information from any other circuit.  In this way, the numerical optimization of circuits for all time-steps for large system simulations could be executed simultaneously on a classical supercomputer, which are regularly used for similar computations.

The volume of the constant-depth circuits grows quadratically with system size $N$, while the depth grows only linearly with $N$.  We emphasize, however, that unlike previous circuit generation techniques, our circuits do not grow in size with increasing numbers of time-steps, but rather remain fixed for a given system size $N$.  This remarkable feature is what enables simulation out to arbitrarily large numbers of time-steps and thus permits long-time dynamic simulations.  Most other methods for circuit generation will produce circuits that grow linearly with increasing numbers of time-steps \cite{wiebe11, smith2019simulating}.  This prohibits dynamic simulations beyond a certain number of time-steps due to the quantum computer encountering circuits that are too large, and thus accumulate too much error due to gate errors and qubit decoherence.

\section{Results}
\label{sec_results}

\begin{figure*}
	\centering
	\includegraphics[scale=0.95]{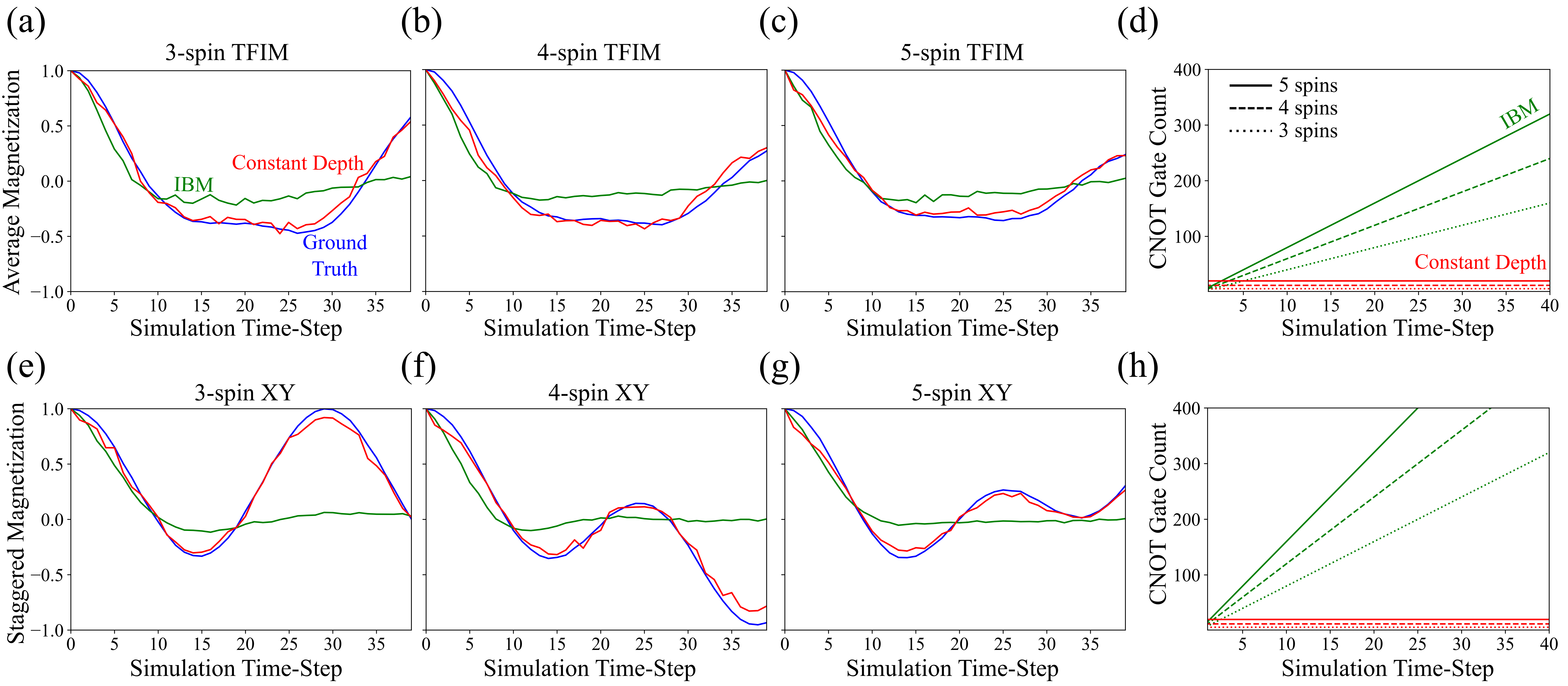}
	\caption{Comparison of simulation results and $\cnot$ gate count for the TFIM and XY model using the constant-depth circuits versus the IBM-compiled circuits.  The top row shows simulation results for a TFIM with 3- (a), 4- (b), and 5-qubits (c).  A noise-free simulator provides the ground truth, shown in blue, with which to compare results from the constant-depth circuits (red) and IBM-compiled circuits (green).  The bottom row shows the analogous simulation results for the XY model with 3- (e), 4- (f), and 5-qubits (g).  (d) and (h) show the number of $\cnot$ gates in the constant-depth (red) and IBM-compiled (green) circuits for each time-step for 3- (dotted line), 4- (dashed line), and 5-qubit (solid line) systems for the TFIM and XY models, respectively.}
	\label{hardware_sim}
\end{figure*}
To demonstrate the power of our constant-depth circuits, we simulate quantum quenches of 3-, 4-, and 5-spin systems defined by the TFIM and XY model on the IBM quantum processor ``ibmq\_athens".  A quantum quench is simulated by initializing the system in the ground state of an initial Hamiltonian, $H_i$, and then evolving the system through time under a final Hamiltonian, $H_f$.  Quenches can simulate a sudden change in a system's environment and provide insights into the non-equilibrium dynamics of various quantum materials.  

To obtain the TFIM, we set $J_y = J_z = 0$ and $\beta = z$ in the Hamiltonian in \cref{Hamiltonian}.  To perform the quench, we assume the external magnetic field is initially turned off, and the qubits are initialized in the ground state of an initial Hamiltonian $H_i(t < 0) = \sum_i -J_x\ \sigma_{i}^{x}\sigma_{i+1}^{x}$, which is a ferromagnetic state oriented along the $x$-axis.  At time $t=0$, a time-dependent magnetic field is instantaneously turned on, and the system evolves under the final Hamiltonian  $H_i(t \ge 0) = -\sum_i \{J_x\sigma_{i}^{x}\sigma_{i+1}^{x} + h_z(t)\sigma_{i}^{z}\}$.  We use parameters from Ref \cite{bassman2020towards}, setting $J_x = 11.83898\,\si{\milli\electronvolt}$ and $h_z(t) = 2J_x \cos(\omega t)$ with $\omega = 0.0048\,\si{\per\femto\second}$, which simulates a simplified model of a $\ch{Re}$-doped mono-layer of $\ch{MoSe2}$ under laser excitation.  A time-step of $3\,\si{\femto\second}$ is used in the simulations. Our observable of interest is the average magnetization of the system along the $x$-axis, given by $m_x(t) = \frac{1}{N}\sum_i \langle\sigma_{i}^{x}(t)\rangle$.

To obtain the XY model, we set $J_z = h_{\beta} =0$ in the Hamiltonian in \cref{Hamiltonian}.  To perform the quench, we assume $J_x = J_y = -1.0\,\si{\electronvolt}$ and let $J_z \rightarrow \infty$, resulting in an initial Hamiltonian $H_i(t<0) = C \sum\sigma_{i}^{z}\sigma_{i+1}^{z}$, where $C$ is an arbitrarily large constant.  The ground state of this Hamiltonian is the N\'eel state, defined as $\ket{\psi_0} = \ket{\uparrow \downarrow \uparrow \cdots \downarrow}$.  At time $t=0$, we instantaneously set $J_z = 0$, giving a final Hamiltonian of $H_f(t \geq 0) = \sum_i\{\sigma_{i}^{x}\sigma_{i+1}^{x} + \sigma_{i}^{y}\sigma_{i+1}^{y}\}$.  A time-step of $0.025\,\si{\femto\second}$ is used in the simulations.  Our observable of interest is the staggered magnetization of the system, which is related to the antiferromagnetic order parameter and given by $m_s(t) = \frac{1}{N}\sum_i (-1)^i \langle\sigma_{i}^{z}(t)\rangle$.

To generate the constant-depth circuits for our simulations, we rely on circuit optimization software provided by the circuit synthesis toolkit BQSKit \cite{BQSKit}.  This suite of software provides several packages which can be used to generate the constant-depth circuits.  The user must provide the matrix representation for the time-evolution operator to be implemented along with the parameter-dependent constant-depth circuit structure.  The circuit synthesis software then proceeds to use numerical optimization to find the optimal parameters for the circuit.  Tutorials including the full code for generating our constant-depth circuits using the BQSKit toolkit are included in the Supplemental Material \cite{github}.

\cref{hardware_sim} shows the simulation results for quenches of the TFIM (top row) and XY model (bottom row) for various system sizes.  In the first three columns, the results from our constant-depth circuits (red) and IBM-compiled circuits (green) are compared to the expected results computed with a noise-free quantum computer simulator (blue).  The $x$-axis gives the number of simulation time-steps, while the $y$-axis gives the time-dependent observable.  Circuits from the IBM compiler, as well as other state-of-the-art general purpose compilers, grow in depth with increasing numbers of time-steps.  For this reason, the IBM-compiled circuits produce qualitatively consistent results for the first few time-steps, but thereafter the circuits are too large, accumulating too much error, to produce reliable results.  A recent benchmark study of dynamic simulations of similar systems on quantum computers found analogous behavior, with high-fidelity results limited to only a handful of time-steps \cite{smith2019simulating}.  In contrast, the results from our constant-depth circuits remain accurate for all time-step counts, and in principle, will remain so out to arbitrarily many time-steps.  These results thus show the power of constant-depth circuits to enable long-time dynamic simulations.

\cref{hardware_sim}d and \ref{hardware_sim}h compare the number of $\cnot$ gates for each time-step in the constant-depth (red) and IBM-compiled (green) circuits for 3- (dotted line), 4- (dashed lined), and 5-spin (solid line) systems.  Clearly the number of $\cnot$ gates remains the same for all time-steps for our constant-depth circuits, but the number grows linearly with increasing numbers of time-steps for the IBM-compiled circuits.  Notice how the number of $\cnot$ gates per time-step for the XY model circuits (\ref{hardware_sim}h) are approximately double the number for the TFIM circuits (\ref{hardware_sim}d), while our constant-depth circuits have the same $\cnot$ count for both models.  

\section{Conclusions}
We have presented a method for generating circuits that are constant in depth with increasing numbers of time-steps for dynamic simulations of quantum materials.  Standard Hamiltonian simulation algorithms produce circuits that grow in depth with increasing time-step count, which limits the number of time-steps that are feasible to simulate on near-term devices. The constant-depth nature of our circuits removes this limit, thereby allowing for simulations to be broken into arbitrarily many time-steps.  This allows for Trotter error to be made negligible by making the time-step arbitrarily small.  Furthermore, the constant-depth circuits enable long-time dynamics to be simulated since simulations can be executed out to arbitrarily long times.  

The constant-depth circuits have a regular structure, which allows for simple construction of the circuits for each system size; for a system of $N$-spins the circuit structure contains only $N(N-1)$ $\cnot$ gates.  These constant-depth circuits are currently applicable to one-dimensional spin models with nearest-neighbor interactions along two or fewer axes, which we denote by $\Hcd$.  They are comprised of two-qubit matchgates acting on nearest neighbor pairs of qubits.  We find that the matchgates for the $\Hcd$ Hamiltonians are special in that they can be decomposed into native-gate circuits requiring only two $\cnot$ gates, as opposed to generic matchgates that require three. The ability to downfold circuits for dynamic simulation under $\Hcd$ Hamiltonians to constant-depth relies on a set of conjectured identities for these special matchgates that we introduce in this paper.  We demonstrate the power of the constant-depth circuits with dynamic simulations of the TFIM and XY models with 3-, 4-, and 5-spins carried out on IBM's quantum processors.  Our results illustrate how the constant-depth circuits enable successful long-time dynamic simulations of quantum materials.  

There are numerous directions for future investigations to explore whether constant-depth circuits can be created for various extensions including two-dimensional models, models with next-nearest neighbor or even longer-range interactions, or full Heisenberg interactions (i.e., interactions along three axes).  Indeed, matchgates have previously been studied for various two-dimensional qubit topologies and for longer-range interactions \cite{brod2012geometries, brod2013computational}.  Paired with incremental improvements in quantum hardware, the ability to extend our constant-depth circuits to more complex systems would pave the way to new discoveries in the behavior of quantum materials by enabling long-time dynamic simulations of systems relevant to scientific and technological problems.

\section*{Acknowledgements}
This work was supported by the U.S. Department of Energy (DOE) under Contract No. DE-AC02-05CH11231, through the Office of Advanced Scientific Computing Research Accelerated Research for Quantum Computing Program (LB, RVB, EY, CI, and WAdJ) and the Advanced Quantum Testbed (ES).

\onecolumngrid
\newpage
\appendix

\section{Quantum circuits for $\boldsymbol\Hcd$ matchgates}
\label{app:matchgates}

\subsection*{Hamiltonian parameter subsets with:
          $\boldsymbol{J_x J_y = J_x J_z = J_y J_z = 0}$ and
          $\boldsymbol{h_\beta = 0}$}

\subsubsection*{Matchgates used in constant-depth circuits for simulating $H = J_x \sum_i \sigma_i^x \sigma_{i+1}^x $}

\vspace{-20pt}
\begin{flalign}
G &:=
{\small\begin{bmatrix}
       \cos\frac{\theta_1}{2} & & &
    -\I\sin\frac{\theta_1}{2} \\
  &    \cos\frac{\theta_1}{2} &
    -\I\sin\frac{\theta_1}{2} \\
  & -\I\sin\frac{\theta_1}{2} &
       \cos\frac{\theta_1}{2} \\
    -\I\sin\frac{\theta_1}{2} & & &
       \cos\frac{\theta_1}{2} \\
\end{bmatrix}} &
{\small\begin{myqcircuit}
  & \ctrl1 & \Rxx1 & \ctrl1 & \qw \\
  & \targ  & \qw   & \targ  & \qw
\end{myqcircuit}} \qquad\qquad\qquad\ \ \
\end{flalign}

\subsubsection*{Matchgates used in constant-depth circuits for simulating $H = J_y \sum_i \sigma_i^y \sigma_{i+1}^y $}

\vspace{-20pt}
\begin{flalign}
G &:=
{\small\begin{bmatrix}
       \cos\frac{\theta_1}{2} & & &
     \I\sin\frac{\theta_1}{2} \\
  &    \cos\frac{\theta_1}{2} &
    -\I\sin\frac{\theta_1}{2} \\
  & -\I\sin\frac{\theta_1}{2} &
       \cos\frac{\theta_1}{2} \\
     \I\sin\frac{\theta_1}{2} & & &
       \cos\frac{\theta_1}{2} \\
\end{bmatrix}} &
{\small\begin{myqcircuit}
  & \Rzp & \ctrl1 & \Rxx1 & \ctrl1 & \Rzm & \qw \\
  & \Rzp & \targ  & \qw   & \targ  & \Rzm & \qw
\end{myqcircuit}} \qquad
\end{flalign}

\subsubsection*{Matchgates used in constant-depth circuits for simulating $H = J_z \sum_i \sigma_i^z \sigma_{i+1}^z $}

\vspace{-20pt}
\begin{flalign}
G &:=
{\small\begin{bmatrix}
      e^{-\frac{\I\theta_2}{2}} \\
  &   e^{ \frac{\I\theta_2}{2}} \\
  &&  e^{ \frac{\I\theta_2}{2}} \\
  &&& e^{-\frac{\I\theta_2}{2}} \\
\end{bmatrix}} &
{\small\begin{myqcircuit}
  & \ctrl1 & \qw   & \ctrl1 & \qw \\
  & \targ  & \Rzz2 & \targ  & \qw
\end{myqcircuit}} \qquad\qquad\qquad\ \ \
\end{flalign}
~

\subsection*{Hamiltonian parameter subsets with:
          $\boldsymbol{J_x J_y = J_x J_z = J_y J_z = 0}$ and
          $\boldsymbol{h_\beta \neq 0}$}

\subsubsection*{$G$ gates used in constant-depth circuits for simulating $H = J_x \sum_i \sigma_i^x \sigma_{i+1}^x + h_x\sum_i \sigma_i^x $}

\vspace{-20pt}
\begin{flalign}
G &:=
\begin{bmatrix}
  \bA & \bC & \bC & \bB \\
  \bC & \bA & \bB & \bC \\
  \bC & \bB & \bA & \bC \\
  \bB & \bC & \bC & \bA
\end{bmatrix} \quad
\begin{aligned}
  \bA &= {\small
            \left[\cos\tfrac{\theta_0}{2}\right]^2 \cos\tfrac{\theta_1}{2} +
         \I \left[\sin\tfrac{\theta_0}{2}\right]^2 \sin\tfrac{\theta_1}{2}} \\
  \bB &= {\small
           -\left[\sin\tfrac{\theta_0}{2}\right]^2 \cos\tfrac{\theta_1}{2} -
         \I \left[\cos\tfrac{\theta_0}{2}\right]^2 \sin\tfrac{\theta_1}{2}} \\
  \bC &= {\small -\tfrac12\I e^{-\frac{\I\theta_1}{2}}\sin\theta_0}
\end{aligned} &
{\small\begin{myqcircuit}
  & \Rxx0 & \ctrl1 & \Rxx1 & \ctrl1 & \qw \\
  & \Rxx0 & \targ  & \qw   & \targ  & \qw
  \gategroup{1}{2}{2}{2}{1.5ex}{--}
\end{myqcircuit}} \qquad\qquad\qquad\ \ \
\end{flalign}

\subsubsection*{$G$ gates used in constant-depth circuits for simulating $H = J_y \sum_i \sigma_i^y \sigma_{i+1}^y + h_y\sum_i \sigma_i^y $}

\vspace{-20pt}
\begin{flalign}
G &:=
\begin{bmatrix}
   \bA &  \bD &  \bD & -\bB \\
  -\bD &  \bA &  \bB &  \bD \\
  -\bD &  \bB &  \bA &  \bD \\
  -\bB & -\bD & -\bD &  \bA
\end{bmatrix} \quad \bD = {\small -\I\bC} &
{\small\begin{myqcircuit}
  & \Ryy0 & \Rzp & \ctrl1 & \Rxx1 & \ctrl1 & \Rzm & \qw \\
  & \Ryy0 & \Rzp & \targ  & \qw   & \targ  & \Rzm & \qw
  \gategroup{1}{2}{2}{2}{1.5ex}{--}
\end{myqcircuit}} \qquad
\end{flalign}

\subsubsection*{Matchgates used in constant-depth circuits for simulating $H = J_z \sum_i \sigma_i^z \sigma_{i+1}^z + h_z\sum_i \sigma_i^z $}

\vspace{-20pt}
\begin{flalign}
G &:=
{\small\begin{bmatrix}
      e^{-\frac{\I\theta_2}{2} - \I\theta_0} \\
  &   e^{ \frac{\I\theta_2}{2}} \\
  &&  e^{ \frac{\I\theta_2}{2}} \\
  &&& e^{-\frac{\I\theta_2}{2} + \I\theta_0} \\
\end{bmatrix}} &
{\small\begin{myqcircuit}
  & \Rzz0 & \ctrl1 & \qw   & \ctrl1 & \qw \\
  & \Rzz0 & \targ  & \Rzz2 & \targ  & \qw
  \gategroup{1}{2}{2}{2}{1.5ex}{--}
\end{myqcircuit}} \qquad\qquad\qquad\ \ \
\end{flalign}

\subsection*{Hamiltonian parameter subsets with: $\boldsymbol{J_x J_y J_z = 0}$ and
                                     $\boldsymbol{h_\beta = 0}$}

\subsubsection*{Matchgates used in constant-depth circuits for simulating $H = \sum_i J_x\sigma_i^x \sigma_{i+1}^x + J_y\sigma_i^y \sigma_{i+1}^y $}

\vspace{-20pt}
\begin{flalign}
G &:=
{\small\begin{bmatrix}
       \cos\frac{\theta_1 - \theta_2}{2} & & &
    -\I\sin\frac{\theta_1 - \theta_2}{2} \\
  &    \cos\frac{\theta_1 + \theta_2}{2} &
    -\I\sin\frac{\theta_1 + \theta_2}{2} \\
  & -\I\sin\frac{\theta_1 + \theta_2}{2} &
       \cos\frac{\theta_1 + \theta_2}{2} \\
    -\I\sin\frac{\theta_1 - \theta_2}{2} & & &
       \cos\frac{\theta_1 - \theta_2}{2} \\
\end{bmatrix}} &
{\small\begin{myqcircuit}
  & \Rxp & \ctrl1 & \Rxx1 & \ctrl1 & \Rxm & \qw \\
  & \Rxp & \targ  & \Rzz2 & \targ  & \Rxm & \qw
  \gategroup{1}{2}{2}{2}{1.5ex}{.}
  \gategroup{1}{4}{1}{4}{1.5ex}{.}
  \gategroup{1}{6}{2}{6}{1.5ex}{.}
\end{myqcircuit}} \qquad
\end{flalign}

\subsubsection*{Matchgates used in constant-depth circuits for simulating $H = \sum_i J_x\sigma_i^x \sigma_{i+1}^x + J_z\sigma_i^z \sigma_{i+1}^z $}

\vspace{-20pt}
\begin{flalign}
G &:=
{\small\begin{bmatrix}
        e^{-\frac{\I\theta_2}{2}}\cos\frac{\theta_1}{2} & & &
    -\I e^{-\frac{\I\theta_2}{2}}\sin\frac{\theta_1}{2} \\
  &     e^{ \frac{\I\theta_2}{2}}\cos\frac{\theta_1}{2} &
    -\I e^{ \frac{\I\theta_2}{2}}\sin\frac{\theta_1}{2} \\
  & -\I e^{ \frac{\I\theta_2}{2}}\sin\frac{\theta_1}{2} &
        e^{ \frac{\I\theta_2}{2}}\cos\frac{\theta_1}{2} \\
    -\I e^{-\frac{\I\theta_2}{2}}\sin\frac{\theta_1}{2} & & &
        e^{-\frac{\I\theta_2}{2}}\cos\frac{\theta_1}{2} \\
\end{bmatrix}} &
{\small\begin{myqcircuit}
  & \ctrl1 & \Rxx1 & \ctrl1 & \qw \\
  & \targ  & \Rzz2 & \targ  & \qw
  \gategroup{2}{3}{2}{3}{1.5ex}{.}
\end{myqcircuit}} \qquad\qquad\qquad\quad
\end{flalign}

\subsubsection*{Matchgates used in constant-depth circuits for simulating $H = \sum_i J_y\sigma_i^y \sigma_{i+1}^y + J_z\sigma_i^z \sigma_{i+1}^z $}

\vspace{-20pt}
\begin{flalign}
G &:=
{\small\begin{bmatrix}
        e^{-\frac{\I\theta_2}{2}}\cos\frac{\theta_1}{2} & & &
     \I e^{-\frac{\I\theta_2}{2}}\sin\frac{\theta_1}{2} \\
  &     e^{ \frac{\I\theta_2}{2}}\cos\frac{\theta_1}{2} &
    -\I e^{ \frac{\I\theta_2}{2}}\sin\frac{\theta_1}{2} \\
  & -\I e^{ \frac{\I\theta_2}{2}}\sin\frac{\theta_1}{2} &
        e^{ \frac{\I\theta_2}{2}}\cos\frac{\theta_1}{2} \\
     \I e^{-\frac{\I\theta_2}{2}}\sin\frac{\theta_1}{2} & & &
        e^{-\frac{\I\theta_2}{2}}\cos\frac{\theta_1}{2} \\
\end{bmatrix}} &
{\small\begin{myqcircuit}
  & \Rzp & \ctrl1 & \Rxx1 & \ctrl1 & \Rzm & \qw \\
  & \Rzp & \targ  & \Rzz2 & \targ  & \Rzm & \qw
  \gategroup{2}{4}{2}{4}{1.5ex}{.}
\end{myqcircuit}} \qquad
\end{flalign}

\subsection*{Hamiltonian parameter subsets with: $\boldsymbol{J_x J_y J_z = 0}$ and
                                     $\boldsymbol{h_\beta \neq 0}$}

\subsubsection*{Matchgates used in constant-depth circuits for simulating $H = J_x\sum_i \sigma_i^x \sigma_{i+1}^x + h_z\sum_i \sigma_i^z$ or $H = J_y\sum_i \sigma_i^y \sigma_{i+1}^y + h_z\sum_i \sigma_i^z$ or $H = \sum_i \{J_x\sigma_i^x \sigma_{i+1}^x + J_y\sigma_i^y \sigma_{i+1}^y\} + h_z\sum_i \sigma_i^z$}

\vspace{-20pt}
\begin{subequations}
\begin{flalign}
G &:=
{\small\begin{bmatrix}
        e^{-\I(\theta_0 + \theta_3)}\cos\frac{\theta_1 - \theta_2}{2} & & &
    -\I e^{ \I(\theta_0 - \theta_3)}\sin\frac{\theta_1 - \theta_2}{2} \\
  &    \cos\frac{\theta_1 + \theta_2}{2} &
    -\I\sin\frac{\theta_1 + \theta_2}{2} \\
  & -\I\sin\frac{\theta_1 + \theta_2}{2} &
       \cos\frac{\theta_1 + \theta_2}{2} \\
    -\I e^{-\I(\theta_0 - \theta_3)}\sin\frac{\theta_1 - \theta_2}{2} & & &
        e^{ \I(\theta_0 + \theta_3)}\cos\frac{\theta_1 - \theta_2}{2} \\
\end{bmatrix}} &
\end{flalign}
\vspace{-7.5pt}
\begin{flalign} &&
{\small\begin{myqcircuit}
  & \Rzz0 & \Rxp & \ctrl1 & \Rxx1 & \ctrl1 & \Rxm & \Rzz3 & \qw \\
  & \Rzz0 & \Rxp & \targ  & \Rzz2 & \targ  & \Rxm & \Rzz3 & \qw
  \gategroup{1}{2}{2}{2}{1.5ex}{--}
  \gategroup{1}{8}{2}{8}{1.5ex}{--}
\end{myqcircuit}} \quad\ \,
\end{flalign}
\end{subequations}

\subsubsection*{$G$ gates used in constant-depth circuits for simulating $H = J_x\sum_i \sigma_i^x \sigma_{i+1}^x + h_y\sum_i \sigma_i^y$ or $H = J_z\sum_i \sigma_i^z \sigma_{i+1}^z + h_y\sum_i \sigma_i^y$ or $H = \sum_i \{J_x\sigma_i^x \sigma_{i+1}^x + J_z\sigma_i^z \sigma_{i+1}^z\} + h_y\sum_i \sigma_i^y$}

\vspace{-20pt}
\begin{subequations}
\begin{flalign}
G &:=
\begin{bmatrix}
        \bE &       \bG &       \bG &      \bF \\
  -\conj\bG &  \conj\bE & -\conj\bF & \conj\bG \\
  -\conj\bG & -\conj\bF &  \conj\bE & \conj\bG \\
        \bF &      -\bG &      -\bG &      \bE
\end{bmatrix} &
{\small\begin{myqcircuit}
  & \Ryy0 & \ctrl1 & \Rxx1 & \ctrl1 & \Ryy3 & \qw \\
  & \Ryy0 & \targ  & \Rzz2 & \targ  & \Ryy3 & \qw
  \gategroup{1}{2}{2}{2}{1.5ex}{--}
  \gategroup{1}{6}{2}{6}{1.5ex}{--}
\end{myqcircuit}} \qquad\qquad\qquad\
\end{flalign}
\vspace{-7.5pt}
\begin{align}
\bE &= {\small
      \tfrac12\left[\cos\tfrac{\theta_1 + \theta_2}{2} +
                    \cos\tfrac{\theta_1 - \theta_2}{2}\cos(\theta_0 + \theta_3)
             \right] -
    \tfrac12\I\left[\sin\tfrac{\theta_1 + \theta_2}{2} -
                    \sin\tfrac{\theta_1 - \theta_2}{2}\cos(\theta_0 - \theta_3)
             \right]} \nonumber\\
\bF &= {\small
      \tfrac12\left[\cos\tfrac{\theta_1 + \theta_2}{2} -
                    \cos\tfrac{\theta_1 - \theta_2}{2}\cos(\theta_0 + \theta_3)
             \right] -
    \tfrac12\I\left[\sin\tfrac{\theta_1 + \theta_2}{2} +
                    \sin\tfrac{\theta_1 - \theta_2}{2}\cos(\theta_0 - \theta_3)
             \right]} \\
\bG &= {\small
        -\tfrac12  \cos\tfrac{\theta_1 - \theta_2}{2}\sin(\theta_0 + \theta_3)
        -\tfrac12\I\sin\tfrac{\theta_1 - \theta_2}{2}\sin(\theta_0 - \theta_3)}
  \nonumber
\end{align}
\end{subequations}

\subsubsection*{$G$ gates used in constant-depth circuits for simulating $H = J_y\sum_i \sigma_i^y \sigma_{i+1}^y + h_x\sum_i \sigma_i^x$ or $H = J_z\sum_i \sigma_i^z \sigma_{i+1}^z + h_x\sum_i \sigma_i^x$ or $H = \sum_i \{J_y\sigma_i^y \sigma_{i+1}^y + J_z\sigma_i^z \sigma_{i+1}^z\} + h_x\sum_i \sigma_i^x$}

\vspace{-20pt}
\begin{flalign}
G &:=
\begin{bmatrix}
        \bE &       \bH &       \bH &      -\bF \\
  -\conj\bH &  \conj\bE & -\conj\bF & -\conj\bH \\
  -\conj\bH & -\conj\bF &  \conj\bE & -\conj\bH \\
       -\bF &       \bH &       \bH &       \bE
\end{bmatrix} \quad \bH = {\small \I\bG} &
{\small\begin{myqcircuit}
  & \Rxx0 & \Rzp & \ctrl1 & \Rxx1 & \ctrl1 & \Rzm & \Rxx3 & \qw \\
  & \Rxx0 & \Rzp & \targ  & \Rzz2 & \targ  & \Rzm & \Rxx3 & \qw
  \gategroup{1}{2}{2}{2}{1.5ex}{--}
  \gategroup{1}{8}{2}{8}{1.5ex}{--}
\end{myqcircuit}} \qquad
\end{flalign}

\section{Proof of \cref{fig:mirror_id}a}
\label{app:proof}

\begin{proof}
Using \eqref{eq:equiv3} recursively, yields
\begin{align}
{\small\begin{myqcircuit}
  & \MG & \qw & \MG & \qw & \qw \\
  & \mg & \MG & \mg & \MG & \qw \\
  & \MG & \mg & \MG & \mg & \qw \\
  & \mg & \qw & \mg & \qw & \qw
\end{myqcircuit}}
\quad &= \quad
{\small\begin{myqcircuit}
  & \qw & \MG & \qw & \MG & \qw & \qw & \qw \\
  & \qw & \mg & \MG & \mg & \qw & \MG & \qw \\
  & \MG & \gw & \mg & \gw & \MG & \mg & \qw \\
  & \mg & \qw & \qw & \qw & \mg & \qw & \qw   \gategroup{1}{3}{3}{5}{2ex}{--}
\end{myqcircuit}} \\[10pt]
 &= \quad
{\small\begin{myqcircuit}
  & \qw & \qw & \MG & \qw & \qw & \qw & \qw \\
  & \qw & \MG & \mg & \MG & \qw & \MG & \qw \\
  & \MG & \mg & \qw & \mg & \MG & \mg & \qw \\
  & \mg & \qw & \qw & \gw & \mg & \gw & \qw   \gategroup{2}{5}{4}{7}{2ex}{--}
\end{myqcircuit}} \\[10pt]
 &= \quad
{\small\begin{myqcircuit}
  & \qw & \qw & \MG & \qw & \qw & \qw \\
  & \qw & \MG & \mg & \MG & \qw & \qw \\
  & \MG & \mg & \MG & \mg & \MG & \qw \\
  & \mg & \qw & \mg & \qw & \mg & \qw
\end{myqcircuit}} \\[10pt]
 &= \quad
{\small\begin{myqcircuit}
  & \qw & \qw & \qw & \MG & \qw & \qw & \qw \\
  & \gw & \MG & \gw & \mg & \MG & \qw & \qw \\
  & \MG & \mg & \MG & \qw & \mg & \MG & \qw \\
  & \mg & \qw & \mg & \qw & \qw & \mg & \qw   \gategroup{2}{2}{4}{4}{2ex}{--}
\end{myqcircuit}} \\[10pt]
 &= \quad
{\small\begin{myqcircuit}
  & \qw & \qw & \gw & \MG & \gw & \qw & \qw \\
  & \MG & \qw & \MG & \mg & \MG & \qw & \qw \\
  & \mg & \MG & \mg & \qw & \mg & \MG & \qw \\
  & \qw & \mg & \qw & \qw & \qw & \mg & \qw   \gategroup{1}{4}{3}{6}{2ex}{--}
\end{myqcircuit}} \\[10pt]
 &= \quad
{\small\begin{myqcircuit}
  & \qw & \MG & \qw & \MG & \qw \\
  & \MG & \mg & \MG & \mg & \qw \\
  & \mg & \MG & \mg & \MG & \qw \\
  & \qw & \mg & \qw & \mg & \qw
\end{myqcircuit}}
\end{align}
\end{proof}

\twocolumngrid
\bibliography{main.bib}

\end{document}